\documentclass[a4paper,11pt]{llncs}

\usepackage{bussproofs}
\usepackage[left=25mm,
 	right=25mm,
 	top=25mm,
 	bottom=25mm]{geometry}
\usepackage{times}
\usepackage{amssymb}
\usepackage{amsmath,bm}
\usepackage{xspace}

\newcommand{\e}[1]{{\bf #1}}
\newcommand{\g}[1]{{\sf #1}}
\newcommand{\Var}[1]{\textit{Var}(#1)}
\newcommand{\DBI}{\e{I}\xspace}
\newcommand{\deps}{\Sigma}
\newcommand{\Dom}{\e{Dom}\xspace}
\newcommand{\schema}{\e{R}\xspace}
\newcommand{\LogSpace}{\textsc{LogSpace}\xspace}
\newcommand{\rul}{\leftarrow}
\newcommand{\incl}{\subseteq}
\newcommand{\stincl}{\subsetneq}
\newcommand{\fact}[1]{\textit{Fact}(#1)}
\newcommand{\sent}[1]{\textit{Sent}(#1)}
\newcommand{\Inst}[1]{\textit{Inst}(#1)}
\newcommand{\chase}[1]{\e{chase}(#1)}
\newcommand{\adom}[1]{\textit{adom}(#1)}
\newcommand{\tup}[1]{(#1)} 

\newcommand{\assg}{\gamma}
\newcommand{\bag}[1]{\{\hspace{-0.07cm}| #1 |\hspace{-0.07cm}\}}
\newcommand{\arunion}[2]{\underset{#1}{\overset{#2}{\bigcup}}}
\newcommand{\normal}[1]{\em #1}
\newcommand{\union}{\cup}
\newcommand{\Nat}{\mathbb{N}}
\newcommand{\set}[1]{\{#1\}}
\newcommand{\chases}[2]{\e{chase}_{#1}(#2)}
\newcommand{\qued}{\hspace*{\fill}$\Box$}
\newcommand{\COUNT}{\e{count}\xspace}
\newcommand{\SUM}{\e{sum}\xspace}
\newcommand{\AVG}{\e{avg}\xspace}

\title{Chases and Bag-Set Certain Answers}
\author{Camilo Thorne}
\institute{
\begin{tabular}{c}
IBM CAS Trento - Trento RISE\\
Piazza Manci 17, 38123, Povo di Trento (Italy)\\
{\tt c.thorne.email@trentorise.eu}
\end{tabular}
}

\begin{document}

\maketitle

\begin{abstract}
In this paper we show that the chase technique
is powerful enough to capture the bag-set semantics of
conjunctive queries over IDBs and
IDs and TGDs. In addition, we argue that in such cases
it provides efficient (\LogSpace) query evaluation
algorithms and that, moreover,
it can serve as a basis for evaluating some restricted
classes of aggregate queries under incomplete information.
\end{abstract}

\section{Introduction}\label{sec1}

One of the main query answering techniques
for unions of conjunctive queries
over incomplete databases
is the chase
\cite{Klug1982B,Vardi1984,Cali2004,ChaseRev2008}.
Incomplete databases (IDBs)
can be seen as
a set \DBI of first order logic (FO) 
ground facts (the database instance) and a set $\deps$ of 
FO  axioms known as integrity constraints or dependencies.
Unions of conjunctive queries (UCQs), on the other hand, are
logical specifications of SQL SELECT-PROJECT-JOIN-UNION queries
belonging to the
positive-existential fragment of FO (with no function symbols).
Answering a query over an incomplete database boils down to
returning the values of all the closed substitutions (or groundings)
under which the query is entailed by the
incomplete database. Such set of values
is known as certain answers \cite{Cali2004,Rosati2007}.
When $\deps$ is a set of inclusion and tuple-generating dependencies
(IDs and TGDs), then the chase is known to provide efficient
(\LogSpace in the size of the data) algorithms
for computing the set of certain answers. 

As we all know, queries in
relational database systems return, 
in practice, bags or multisets of answers.
This is because it can be quite costly to delete
repeated answers \cite{Vardi1993,Abiteboul1995}. Arguably,
this property still holds for incomplete databases
that integrate relational datasources
The main contributions
of this paper are thus the following:
\textit{(i)} extending the notion of certain answers
to cover bags and \textit{(ii)} showing that
chases are powerful enough to
capture these bags of certain answers when
evaluating a UCQ over an IDB with IDs and TGDs. These results
outline algorithms that, on the one hand, efficiently compute these
bags of certain answers and, on the other hand, can be extended
to more expressive classes of queries.

\section{Preliminaries}

Let $\schema := \set{R_{1},...,R_{n}}$ 
be a set of relation symbols each with
an associated
arity (a positive integer),
called
{\em database schema}. Let
\Dom be a countably infinite set of constants
called {\em domain}. 
A {\em tuple} over \Dom is any finite
sequence $\vec{c}$ of domain constants.
A {\em fact} over \schema and \Dom is a relational atom
over \schema and \Dom.
It is said to be {\em ground} when it contains no variables and
{\em positive} when, in addition, it is not negated.
We will denote by $\fact{\schema}$ the set of all positive facts.
Furthermore, we will denote by $\sent{\schema}$ the set of all FO sentences
(closed formulas) over \schema and \Dom.
Notice that $\fact{\schema} \stincl \sent{\schema}$.

A {\em union of conjunctive queries} $q$ (UCQ)
of arity $n$
is a rule over $\schema$ of the form 
$
q(\vec{x}) \rul
\phi_{0}(\vec{x},\vec{y}_0,\vec{c}_0) \lor ... \lor 
\phi_{k}(\vec{x},\vec{y}_k,\vec{c}_k)
$ 
where $q(\vec{x})$ is called the query's {\em head}, 
$\vec{x}$ is a sequence of $n \geq 0$
{\em distinguished variables} and
$\phi_{0}(\vec{x},\vec{y}_0,\vec{c}_0) \lor ... \lor 
\phi_{k}(\vec{x},\vec{y}_k,\vec{c}_k)$ is the
the query's {\em body}, where
the $\phi_i$s, for $i \in [0,k]$, are conjunctions of atoms.
A UCQ is said to be {\em boolean} if $\vec{x}$
is an empty sequence.

A {\em database instance} or, simply, {\em database} (DB),
of \schema is a finite set 
$\DBI \stincl \fact{\schema}$
of ground facts. The {\em active domain} of \e{I}, $\adom{\DBI}$,
denotes the (finite) set
of constants that occur among the facts in \DBI.
We will denote by $\Inst{\schema}$ the class of all such DBs.
An {\em incomplete database} (IDB) of \e{R} is a pair 
$
\tup{\deps, \DBI}
$ 
where $\deps \incl \sent{\schema}$ is a finite set of
FO axioms called {\em dependencies} or {\em integrity constraints}
and \DBI is a DB of \schema.

A {\em grounding} over UCQ $q$ is a function
$\assg \colon \Var{q} \to \Dom$,
where $\textit{Var}(q)$ denotes the set of variables 
of $q$. 
Groundings can be extended to sequences of variables,
FO atoms and FO formulas in the usual way. We will
denote by $q\assg$ the grounding of the body of a UCQ $q$ by 
$\assg$.
The set of all
such groundings will be denoted $^{\Var{q}}\Dom$.

A grounding $\assg$ is said to {\em satisfy} a UCQ $q$ w.r.t. a
DB \e{I} iff $\DBI \models q\assg$, i.e., if \DBI is a model
(in FO terms) of $q\assg$, the grounded body of $q$. It is said to
{\em satisfy} $q$ w.r.t. to an IDB $\tup{\deps,\DBI}$
iff $\deps \union \DBI \models q\assg$, i.e., if $\deps \cup \DBI$ entails
(in \g{FO} terms) $q\assg$  \cite{Abiteboul1995,Cali2004,Rosati2007}. 

\section{Bag-Set Semantics and Chases}

Let $X$ be a set. A {\em bag}
or {\em multiset} $B$ (over $X$) is a function $B \colon X \to \Nat \union \set{\infty}$.
We say that an $x \in X$ {\em belongs} to $B$
whenever $B(x) \geq 1$, where the integer $B(x)$ is called the
{\em multiplicity} of $x$ in $B$. 
As with sets, bags will be denoted by extension or intension using the special
brackets $\bag{\cdot}$.

\begin{definition}
[\e{Bag-set answers}]
Let \e{I} be a DB and
$q$ a CQ of arity $n$ over a schema \e{R}. We define the {\em bag-set
answers} of $q$ as the bag:
$$
q[\DBI] := \bag{\vec{c} \in \adom{\DBI}^n \mid \vec{c} = \assg(\vec{x}) \text{ for some } \assg \text{ s.t. } 
\DBI \models q\assg}. 
$$
\end{definition}

\begin{definition}
[\e{Bag-set certain answers}]
Let $\tup{\deps, \DBI}$ be an IDB and
$q$ a UCQ of arity $n$ over a schema \e{R}. We define the {\em bag-set certain
answers} of $q$ as the bag:
$$
q[\deps, \DBI] := \bag{\vec{c} \in \adom{\DBI}^n \mid \vec{c} = \assg(\vec{x}) \text{ for some } \assg \text{ s.t. } 
\deps \cup \DBI \models q\assg}. 
$$
\end{definition}

Notice that
by answers we understand the tuples returned by the satisfying groundings.
Notice, too, that in the first case we check for groundings that are satisfying
over {\em one} single DB and in the second for groundings
that are satisfying for a whole space of DBs, namely, {\em all}
the models of the IDB. 

Chases generate DBs (a.k.a. {\em canonical} or {\em universal} DBs) 
that minimally complete the information held by an IDB
\cite{ChaseRev2008}.
Chases are linked to dependencies. In what follows we will consider
only two general kinds of such dependencies, namely, 
{\em inclusion dependencies} (IDs), which are axioms 
of the form $\forall \vec{x} (R(\vec{x}) \to R'(\vec{x}))$,
and {\em tuple-generating dependencies} (TGDs), which are
axioms of the form 
$\forall \vec{x} (R(\vec{x}) \to \exists \vec{y} R'(\vec{x}, \vec{y}))$,
as is common in IDB literature 
\cite{Abiteboul1995,Cali2004,Klug1982B}.

\begin{definition}
\e{(Chase rules)}
Let $\deps$ 
be a (finite) set of IDs and TGDs. The set $R$ of {\em chase rules}
generate new facts by applying the dependencies $\rho \in \deps$
to previously computed facts:
\begin{center}
\small{
\begin{tabular}{c}
\AxiomC{$f := R(\vec{c})$}
\AxiomC{$\rho := \forall \vec{x} (R(\vec{x}) \to R'(\vec{x}))$}
\LeftLabel{\normal ID}
\BinaryInfC{$f_{new} := R'(\vec{c})$}
\DisplayProof\\
\,\\
\AxiomC{$f := R(\vec{c})$}
\AxiomC{$\rho := \forall \vec{x} (R(\vec{x}) \to \exists \vec{y} R'(\vec{x}, \vec{y}))$}
\LeftLabel{\normal TGD}
\BinaryInfC{$f_{new} := R'(\vec{c}, \vec{c}_{new})$}
\DisplayProof
\end{tabular}
}
\end{center}
Where $\vec{c}_{new}$ is a sequence of fresh constants from \Dom. 
Note that there will be as many rules as dependencies in $\deps$. We will
denote by $R$ the set of all chase rules (for a given set $\deps$ of 
integrity constraints). Notice that rules in $R$ are in one-to-one
correspondence with dependencies in $\Sigma$.
\end{definition}

We say that a fact $f$ is {\em derivable from
$\deps$ and \DBI modulo} $R$ and write $\deps \cup \DBI \vdash_{R} f$
whenever \textit{(i)} $f \in \DBI$ or \textit{(ii)} 
there exists a finite sequence $f_0,...,f_n$ such that 
$f_n = f$ and for each $i \in [0,n]$, either
$f_i \in \DBI$ or there exists
$j \leq i$ s.t. $f_i$ follows from $\deps$ and $f_j$ by some
rule in $R$.
The chase rules are sound and complete 
w.r.t. (standard FO) entailment, in the sense that, for
all facts $f \in \fact{\schema}$,
$\deps \union \DBI \vdash_{R} f$ 
iff $\deps \union \DBI \models f$.

\begin{definition}
\e{(Chase)}
Let $\tup{\deps, \DBI}$ be an IDB
over \schema with $\deps$ a set of IDs and TGDs. 
The {\em chase} of $\tup{\deps, \DBI}$, 
denoted $\chase{\deps,\DBI}$,
is defined inductively
as follows:
\begin{itemize}
\item $\chases{0}{\deps,\DBI} := \DBI$.
\item $\chases{i+1}{\deps,\DBI} := \chases{i}{\deps,\DBI} \union \{f_{new}\}$.
\item $\chase{\deps,\DBI} := \arunion{i \in \Nat}{}\chases{i}{\deps,\DBI}$.
\end{itemize}
Where $f_{new}$ follows from $\chases{i}{\deps,\DBI}$
{\em modulo} $R$. The integer 
$i \in \Nat$ is called the {\em level} of the chase. Notice that
we are iteratively generating a (potentially infinite) instance of \schema
from ground atoms contained by \DBI. Furthermore, 
for each $i \in \Nat$, $\chases{i}{\deps,\DBI}$ is an instance of
\schema.
\end{definition}

Chases can also be seen (in FO), 
as the deductive closure $(\deps \union \DBI)^{\vdash_{R}}$
of the IDB (seen as a FO axiomatisation) w.r.t. the set
$\fact{\schema}$ of atomic facts over $\schema$.
Formally, $(\deps \union \DBI)^{\vdash_{R}} := 
\{f \in \fact{\schema} \mid \deps \union \DBI
\vdash_{R} f \text{ implies } f \in  
(\deps \union \DBI)^{\vdash_{R}} \}$. 

\begin{proposition}
Let $\tup{\deps, \DBI}$ 
be an IDB over \e{R}, where $\deps$ is a set of 
IDs and TGDs. Then, $(\deps \union \DBI)^{\vdash_{R}} = \chase{\deps,\DBI}$.
\end{proposition}

\begin{theorem}\label{cert}
Let $\tup{\deps, \DBI}$ 
be an IDB over \e{R}, with $\deps$ a set of 
IDs and TGDs. Let $q$ be a UCQ and
$\assg$ an arbitrary grounding. Then,
$
\chase{\deps, \DBI} \models q\assg \text{ iff } \Sigma \union
\DBI \models q\assg.
$
\end{theorem}

\begin{proof} 
The sense ($\Leftarrow$) is immediate, since $\deps \union
\DBI \incl (\deps \union
\DBI)^{\vdash_{R}} 
= \chase{\deps, \DBI}$. For
($\Rightarrow$), obeserve that
$(\deps \union \DBI)^{\vdash_{R}} = 
\chase{\deps, \DBI}$. Therefore, there
exists a finite sequence of applications of chase rules starting from
$\tup{\deps, \DBI}$ and ending, at some level $k$ of the chase, with all the facts in 
$q\assg$. But then,
since chase rules are sound w.r.t. entailment, $\deps \union
\DBI \models q\assg$.\qued
\end{proof}

\begin{lemma}\label{assgn}
Let $\tup{\deps, \DBI}$ 
be an IDB and $q$ a UCQ
Then:
\begin{equation*}
\begin{array}{rc}
\set{\assg \in \-^{\Var{q}}\Dom \mid \deps \union \DBI \models q\assg} 
& = \\
\set{\assg \in \-^{\Var{q}}\Dom \mid \chase{\deps, \DBI}\models q\assg}
\end{array}
\end{equation*}
\end{lemma}

\begin{proof}
Let $\assg \in \-^{\Var{q}}\Dom$. By Theorem~\ref{cert},
$\deps \union \DBI \models q\assg$ iff $\chase{\deps, \DBI}\models q\assg$,
whence the result. \qued
\end{proof}

\begin{theorem}\label{equiv}
Let $\tup{\deps, \DBI}$ 
be an IDB with $\deps$ a set of 
IDs and TGDs. Let $q$ be a CQ
of arity $n$ and distinguished variables $\vec{x}$. Then it holds that:
\begin{equation*}
\begin{array}{rc}
\bag{\assg(\vec{x}) \in \adom{\DBI}^n \mid \assg \text{ s.t. } \deps \union \DBI \models q\assg}
& = \\
\bag{\assg(\vec{x}) \in \adom{\DBI}^n \mid \assg \text{ s.t. } \chase{\deps, \DBI} \models q\assg}
& 
\end{array}
\end{equation*}
\end{theorem}

\begin{proof}
Put, for brevity, 
$B := \bag{\assg(\vec{x}) \in \adom{\DBI}^n \mid \assg \text{ s.t. } 
\chase{\deps, \DBI} \models q\assg}$ and
$B' := \bag{\assg(\vec{x}) \in \adom{\DBI}^n \mid \assg \text{ s.t. } 
\deps \union \DBI \models q\assg}$.
Assume now for contradiction that there exists a tuple $\vec{c}$ s.t., w.l.o.g., $B(\vec{c}) >
B'(\vec{c})$ (the other case is analogous). 
This implies that for some pair of groundings $\assg, \assg'$
from $\Var{q}$ to \Dom:
\begin{itemize}
\item Either $\assg = \assg'$. Then, since, by Lemma~\ref{assgn}
the set of groundings $\assg \colon \Var{q} \to \Dom$ s.t.
$\chase{\deps, \DBI} \models q\assg$
or $\deps \union \DBI \models q\assg$
are exactly the same, it should be the case that
$\vec{c}$ occurs in $B$ and $B'$ with the same multiplicities,
which is impossible.
\item Or $\assg \neq \assg'$. In which case it holds that:
\begin{itemize}
\item[\textit{(i)}] $\assg(\vec{x}) = \assg'(\vec{x}) := \vec{c}$ and
$\assg$ and $\assg'$ are satisfying assignments for
$q$ and $\chase{\deps, \DBI}$.
\item[\textit{(ii)}] either $\assg$ or $\assg'$ is satisfying for $q$ and
$\deps \union \DBI$, but not both.
\end{itemize}
Assume that the non satisfying
grounding is $\assg$ and suppose, moreover, that no such other grounding $\assg''$
behaves like $\assg$ or $\assg'$. Following, again,
Lemma~\ref{assgn}, if it is the case that $\chase{\deps, \DBI} \models q\assg$,
then $\deps \union \DBI \models q\assg$, i.e., $\assg$ is a satisfying
grounding for this entailment. But this is impossible. \qued
\end{itemize}
\end{proof}

\begin{corollary}\label{red-chase}
Let $\tup{\deps, \DBI}$ 
be an IDB with $\deps$ a set of 
IDs and TGDs. Let $q$ be a UCQ. Then,
$
q[\deps, \DBI] = q[\chase{\deps, \DBI}].
$
\end{corollary}

Corrollary~\ref{red-chase} provides an efficient way of
computing $q[\deps,\DBI]$. 
We can apply
results from description logics expressive enough to capture
IDs and TGDs \cite{Calvanese2007C}. We can use a so-called
{\em perfect rewriting} algorithm, $\textsc{PerfectRef}(\cdot,\cdot,\cdot)$,
such that $\textsc{PerfectRef}(q,\deps,\DBI) = q[\chase{\deps, \DBI}]$.
What a perfect rewriting does is to, so to speak, "compile"
the dependencies in $\deps$ into the UCQ $q$ by rewriting
$q$'s body accordingly (UCQs are closed under such rewritings). Next,
we can evaluate this (expanded) UCQ over \DBI, thus reducing
query evaluation over IDBs to query evaluation over DBs.
Since this does not affect the {\em data complexity} of query evaluation, 
i.e., the complexity of
evaluating $q$ over IDB $\tup{\deps,\DBI}$ when measured only
w.r.t. the number of tuples occuring in \DBI, we get
\LogSpace data complexity -- the complexity of query evaluation
over plain DBs with SQL queries \cite{Abiteboul1995,Vardi1982}.

\section{Conclusions and Related Work}\label{sec5}

We have showed that chases can be used to compute
not only sets of certain answers, but also their bags.
This is because when we apply chase rules, we reason basically
over the satisfying groundings, and collect later, in a set, the
bindings $\vec{c}$ of an UCQ $q$'s distinguished variables
$\vec{x}$. Chases are also sound and complete w.r.t.
bag-set certain answers when IDBs make use only of IDs and TGDs.

A particularly interesting field of application for these bag chase
techniques are SQL aggregate queries (AQs), e.g.,
\COUNT, \SUM, or \AVG-queries with no nested conditions
\cite{Cohen2007,Abiteboul1995}, in incomplete information settings.
The semantics of those queries in the relational DB setting typically
involves aggregates (bags) as arguments of the aggregation functions.
In \cite{Thorne2008B} we show that a semantics for AQs 
over ontologies and IDBs can be achieved by
returning first the certain answers over
$\tup{\deps, \DBI}$ of an 
{\em auxiliary} UCQ $q_{\textit{aux}}$ associated to
the AQ $q$ and aggregating later
over those certain answers, yieldng a unique value. 
This constrasts with \cite{Kolaitis2008}, where
chases applied to AQs yield intervals of (dense) values.
The results of this paper imply that this can be done efficiently and 
without deleting duplicates using, 
say, $\textsc{PerfectRef}(\cdot,\cdot,\cdot)$, when
$\deps$ is a set of IDs and TGDs. 

\bibliographystyle{plain}
\bibliography{chases}

\begin{thebibliography}{10}

\bibitem{Abiteboul1995}
Serge Abiteboul, Richard Hull, and Victor Vianu.
\newblock {\em Foundations of Databases}.
\newblock Addison-Welsey, 1995.

\bibitem{Cali2004}
Andrea Cali, Diego Calvanese, Giuseppe~De Giacomo, and Maurizio Lenzerini.
\newblock Data integration under integrity constraints.
\newblock {\em Information Systems}, 29(2):147--163, 2004.

\bibitem{Calvanese2007C}
Diego Calvanese, Giuseppe~De Giacomo, Domenico Lembo, Maurizio Lenzerini, and
  Riccardo Rosati.
\newblock Tractable reasoning and efficient query answering in description
  logics: The $\dllite$ family.
\newblock {\em Journal of Automated Reasoning}, 39(3):385--429, 2007.

\bibitem{Thorne2008B}
Diego Calvanese, Werner Nutt, Evgeny Kharlamov, and Camilo Thorne.
\newblock Aggregate queries over ontologies.
\newblock In {\em Proceedings of the 2nd International Workshop on Ontologies
  and Information Systems for the Semantic Web ({ONISW 2008})}, 2008.

\bibitem{Vardi1993}
Surajit Chaudhuri and Moshe Vardi.
\newblock Optimization of real conjunctive queries.
\newblock In {\em Proceedings of the 12th {ACM SIGACT-SIGMOD} Symposium on
  Principles of Database Systems ({PODS 1993})}, 1993.

\bibitem{Cohen2007}
Sara Cohen, Wernet Nutt, and Yeshoshua Sagiv.
\newblock Deciding equivalences among conjunctive aggregate queries.
\newblock {\em Journal of the ACM}, 54(2):1--50, 2007.

\bibitem{ChaseRev2008}
Alin Deutsch, Alan Nash, and Jeff Remmel.
\newblock The chase revisited.
\newblock In {\em Proceedings of the 27th ACM SIGMOD-SIGACT-SIGART Symposium on
  Principles of Database Systems ({PODS 2008})}, 2008.

\bibitem{Klug1982B}
Anthony Klug and David Johnson.
\newblock Testing containment of conjunctive queries under functional and
  inclusion dependencies.
\newblock In {\em Proceedings of the 1st {ACM SIGACT-SIGMOD} Symposium on
  Principles of Database Systems ({PODS 1982})}, 1982.

\bibitem{Kolaitis2008}
Phokion Kolaitis and Foto Afrati.
\newblock Answering aggregate queries in data exchange.
\newblock In {\em Proceedings of the 27th {ACM SIGACT-SIGMOD} Symposium on
  Principles of Database Systems ({PODS 2008})}, 2008.

\bibitem{Rosati2007}
Riccardo Rosati.
\newblock The limits of querying ontologies.
\newblock In {\em Proceedings of the Eleventh International Conference on
  Database Theory ({ICDT 2007})}, 2007.

\bibitem{Vardi1982}
Moshe Vardi.
\newblock The complexity of relational query languages.
\newblock In {\em Proceedings of the Fourteenth Annual {ACM} Symposium on
  Theory of Computing}, 1982.

\bibitem{Vardi1984}
Moshe Vardi and Catriel Beeri.
\newblock A proof procedure for data dependencies.
\newblock {\em Journal of the ACM}, 31(2):718--741, 1984.

\end{thebibliography}

\end{document}